\documentclass[11pt,a4paper]{article}
\usepackage{amssymb}
\usepackage{fullpage}
\newtheorem{theorem}{Theorem}[section]
\newtheorem{lemma}[theorem]{Lemma}
\newtheorem{claim}[theorem]{Claim}
\newtheorem{corollary}[theorem]{Corollary}

\newtheorem{definition}{Definition}

\newcommand{\qed}{\hfill $\Box$ \bigbreak}
\newenvironment{proof}{\noindent{\bf Proof.~}}{\qed}

\def\cD{{\cal D}}
\def\cF{{\cal F}}
\def\cG{{\cal G}}
\def\cM{{\cal M}}
\def\cU{{\cal U}}

\begin{document}
%%%%%%%%%%%%%%%%%%%%%%%%%%%%%%%%%%%%%%%

\title{Compact Ancestry Labeling Schemes for Trees of Small Depth\thanks{This research is supported in part by the ANR project ALADDIN, by the INRIA
project GANG, and by COST Action 295 DYNAMO.}}

\author{
Pierre Fraigniaud\\[1ex]
{\small CNRS and Univ. Paris Diderot}\\{\small\sl pierre.fraigniaud@liafa.jussieu.fr}
%\thanks{CNRS and Universit\'e Paris Diderot - Paris 7, France.
%E-mail: \texttt{Pierre.Fraigniaud@liafa.jussieu.fr}.}
\and
Amos Korman\\[1ex]
{\small CNRS and Univ. Paris Diderot}\\{\small\sl amos.korman@liafa.jussieu.fr}
%\thanks{CNRS and Universit\'e Paris Diderot - Paris 7, France.
%E-mail: \texttt{Pandit@liafa.jussieu.fr}.}
}

\date{}

\maketitle

\begin{abstract}
An {\em ancestry labeling scheme} labels the nodes of any tree in such a way
that ancestry queries between any two nodes in a tree can be answered
%in constant time
just by looking at their corresponding labels. The common measure to evaluate the quality of an
ancestry labeling scheme is by its {\em label size}, that is the maximal number of bits
stored in a label, taken over all $n$-node trees.
The design of ancestry labeling schemes finds applications in XML search engines. In the context of these
applications, even small improvements in the label size are important. In fact, the literature about this topic is interested in the  exact label size rather than just its order of magnitude.
As a result, following the proposal of an original scheme of size $2\log n$ bits,
 a considerable amount of work was devoted to improve the bound on the label size.
The current state of the art upper bound is $\log n + O(\sqrt{\log n})$ bits which
is still far from the known $\log n + \Omega(\log\log n)$ lower bound. Moreover, the hidden
constant factor in the additive $O(\sqrt{\log n})$ term is large,
which makes this term dominate the label size for typical current
XML trees.

In attempt  to provide  good performances for real XML data, we rely on
the observation that the depth of a typical XML tree is bounded from above by a small constant. Having this in mind,
we present an ancestry labeling scheme of size $\log n+2\log d +O(1)$, for the family of
trees with at most $n$ nodes and depth at most $d$.
% This scheme is practical for typical XML trees.
In addition to our main result,
we
prove a result that may be of independent interest concerning the existence of a linear {\em universal graph}
for the family of  forests with trees of bounded depth.
\end{abstract}

\newpage

%%%%%%%%%%%%%%%%%%%%%%%%%%%%%%%%%%%%%%%
\section{Introduction}
\label{section:Introduction}
%%%%%%%%%%%%%%%%%%%%%%%%%%%%%%%%%%%%%%%

\subsection{Background}
\label{section:Background}

It is often the case that when people wish to
 retrieve data from the Internet, they use search engines like Yahoo or Google which
provide full-text indexing services (the user gives some keywords and the engine returns documents containing these keywords).
In contrast to such search engines, the evolving XML Web-standard \cite{ABS99,xml}
aims for allowing more sophisticated queries of documents. By describing the semantic structure
of the document components, it allows users to not only ask
full-text  queries (find documents containing the phrase ``computer science researches'') but also ask for more sophisticated data
(find all items about computer science researches that did their Phd at ETH Z\"urich and whose age is below 35).

To implement such sophisticated queries,
Web documents obeying the XML standard are viewed as labeled trees, and
typical queries over the documents amount to
testing relationships between document items, which correspond to ancestry queries
among the corresponding tree nodes \cite{ABS99,DFFLD99,xsl,xslt}. To process such queries,
XML query engines often use an index structure, typically a big hash
table, whose entries are the tag names in the indexed documents. Due to the enormous size
of the Web data and to its distributed nature, it is essential to answer queries
using the index labels only, without accessing the actual documents.
To allow good performances, it is essential that a large portion of the index structure
resides in the main memory. Since we are dealing here with a huge number of index labels,
reducing the length of the label size, even by a constant factor,
is critical for the reduction of memory cost and for performance improvement. For more details
regarding XML search engines see, e.g., \cite{AAKMT01,AKM01,CKM02}.

Labeling schemes which are currently being used by
actual systems are variants of the following interval based ancestry labeling scheme \cite{KNR92,SK85}.
Enumerate the leaves from left to right and label each node $u$ by the interval $[\ell_s,\ell_l]$,
where $\ell_s$ (respectively, $\ell_l$) is the smallest (resp., largest) leaf descendant of $u$.
An ancestry query then amounts to an interval containment query between the corresponding interval labels.
It is easy to see that the size of the labels produced by this simple scheme is bounded by $2\log n$ bits,
where $n$ is the size of the tree.

A considerable amount of research was devoted to improve the upper bound on the label size as much as possible \cite{AAKMT01,AKM01,TZ01}.
The current state of the art upper bound \cite{AAKMT01} is $\log n + O(\sqrt{\log n})$ which
is still far from the known $\log n + \Omega(\log\log n)$ lower bound \cite{ABR05}. Moreover, the hidden
constant factor in the additive $O(\sqrt{\log n})$ term is large,
which makes this term dominate the label size in the average size of current
applications. Following that work, \cite{KMS02} suggested other ancestry labeling schemes whose worst case bound is $1.5\log n +O(1)$
but perform better than the scheme of \cite{AAKMT01} for typical XML data.

In attempt to provide good performances for real XML instances, we rely on
the observation that a typical XML tree has extremely small depth (cf. \cite{CKM02,MTP06,MBV05}).
For example, by examining  about 200,000
XML documents on the Web, Mignet et al.~\cite{MBV05}
found  that the average depth of an XML tree
is 4, and that 99\% of the trees have depth at most 8. Motivated by this observation,
we concentrate on bounded depth trees, and prove an upper bound of
 $\log n+2\log d +O(1)$ for the size of an ancestry labeling scheme for the family of
$n$-node trees whose depth  is bounded by $d$. (In fact, our bound holds even for forests
rather than just for trees.)

It is not clear whether one can adapt the techniques from previous schemes to
perform better on trees of small depth. For example, the simple interval scheme
has label size $2\log n$ also for trees with constant depth. As another example, before starting
the actual labeling process, the ancestry scheme in
\cite{AAKMT01} first transforms the given  tree to a binary tree.
This transformation already results with a tree of depth $\Omega(\log n)$, even if the given tree has constant depth.
Moreover, previous relevant schemes extensively use and rely on a specific technique,
for using {\em alphabetic codes} on different subpaths. This technique, at least on its surface, does not seem to
be more effective on short subpaths, than on long ones.

In contrast, this paper uses a different  technique that does not rely on alphabetic codes.
Informally, the idea behind our scheme is the following.
The labels of the nodes are taken from a small set of integers $U$, thus ensuring short labels.
Each integer in $U$ is associated with
 some interval taken from some limited range.
The fundamental rule of our labeling scheme is that a node $u$ is an ancestor of $v$ if and only if  the interval associated
with the label of $u$ (i.e., the corresponding integer in $U$) contains the
interval associated with the label of $v$.
That way, the ancestry query can be answered very easily, simply by comparing the corresponding intervals.
The main technical challenge is to find a way to define and nest these intervals between themselves to
  be able to appropriately map the nodes of any $n$-node forest of bounded depth into $U$, while keeping $U$ small.

\subsection{Other related work}

Implicit labeling schemes were first
introduced in \cite{KNR92}, where an elegant adjacency labeling schemes
of size $2\log n$ is established on  $n$-node trees.
That paper also notices a relation between adjacency labeling schemes and \emph{universal graphs} (see also \cite{AR02+,GP01a,KPR06}).
Precisely, it is shown  that there exists an adjacency labeling scheme with label size $k$ for a graph family $\cG$
if and only if there exists a universal graph for $\cG$ with $2^k$ nodes.

Adjacency labeling schemes on trees were further investigated in an
attempt to reduce the constant factor in the label size. In \cite{KM01}  an
adjacency labeling scheme  using label size of $\log n+O(\sqrt{\log
n})$ is presented; and in \cite{AR02+} the label size was further reduced to $\log
n+O(\log^* n)$. This current state of the art bound implies the existence
of a universal graph for the family of $n$-node trees with $2^{O(\log^*(n))}n$ nodes.

Labeling schemes were also proposed for other  decision problems on graphs,
including distance
\cite{ABR05,GP01a,GPPR01,GKKPP01,KM01,KPR06,P99:lbl,T01},
routing \cite{FG01,TZ01}, flow \cite{KK06,KKKP04}, vertex connectivity \cite{K07,KKKP04}, nearest common ancestor
\cite{AGKR01,Peleg00:lca}, and various other tree functions, such as
center, separation level, and Steiner weight of a given subset of
vertices \cite{Peleg00:lca}. See \cite{GP01b} for a survey on static labeling schemes.
Dynamic labeling schemes were investigated in a number of papers, e.g., \cite{K05,K08,KPR04,KP03}.

\subsection{Our contributions}
We present an ancestry
labeling scheme of size $\log n+2\log d +O(1)$ for the family of rooted
forests with at most $n$ nodes and depth at most $d$.
Our result is
essentially optimal for rooted trees with constant depth, and thus for the typical XML trees.

As a corollary of our main theorem, we get an adjacency scheme of size $\log n+3\log d +O(1)$ for
the family of forests with at most $n$ nodes and depth bounded by $d$. This, in particular, implies
the existence of a linear universal graph
for the family of  forests with constant depth. Namely, we show the existence
of a  graph of size $O(n)$ that contains all $n$-node forests of constant depth
as vertex induced subgraphs.

%%%%%%%%%%%%%%%%%%%%%%%%%%%%%%%%%%%%%%%
\section{Preliminaries}
%%%%%%%%%%%%%%%%%%%%%%%%%%%%%%%%%%%%%%%

Let $T$ be a tree rooted at some node $r$ referred as the {\em root} of $T$.
 The {\em depth} of a node $u\in V(T)$
is defined as 1 plus the hop distance from $u$ to the root of $T$.
In particular, the depth of the root is 1.
The depth of $T$ is the maximum depth of a node in $T$.
Let $u$ and $v$ be two nodes in $T$.
We say that $u$ is an {\em ancestor} of $v$ if $u\neq v$ and $u$ is one of the nodes on the shortest path connecting
$v$ and the root of $T$.

A {\em rooted forest} $F$ is a collection of rooted trees. The depth of $F$ is
the maximum depth of a tree in $F$.
For two nodes $u$ and $v$  in $F$,
we say that $u$ is an  ancestor of $v$ if and only if $u$ is an  ancestor of $v$ in one of the trees in $F$.
For integers $n$ and $d$, let $\cF(n,d)$ denote the family of all rooted forests with at most $n$ nodes
and depth bounded from above by $d$.

An {\em ancestry labeling
scheme} $( \cM,\cD )$ for a family of rooted forests $\cF$ is
composed of the following components:
\begin{enumerate}
\item A {\em marker} algorithm $\cM$ that, given
a forest $F$ in $\cF$,
assigns labels to its nodes.
\item A polynomial time {\em
decoder} algorithm $\cD$ that given two labels $\ell_1$ and $\ell_2$ in the output domain of $\cM$, returns a boolean in $\{0,1\}$.
\end{enumerate}
These components must satisfy that if $L(u)$ and $L(v)$ denote the labels assigned by the marker to
two nodes $u$ and $v$ in some rooted forest $F\in\cF$, then
\[
\cD(L(u),L(v))=1 \iff \mbox{$u$ is an ancestor of $v$ in $F$.}
\]
It is important to note that the decoder $\cD$ is independent of the forest $F$. Thus $\cD$ can be
viewed as a method for computing ancestry values in a ``distributed''
fashion, given any pair of labels and knowing that the forest belongs
to some specific family $\cF$.

The common complexity measure used to evaluate a labeling scheme
$(\cM,\cD)$ is the {\em label size}, that is
the maximum number of bits in a label assigned by the marker
algorithm $\cM$ to any node in any forest in $\cF$.

Given two integers $a$ and $b$, where $a<b$, let $[a,b]$ (respectively, $[a,b)$) denote
the interval containing the integers $i$ such that $a\leq i\leq b$ (resp., $a\leq i< b$).
Given a graph $G$,  let $|G|$ denote
the number of nodes in $G$.

%%%%%%%%%%%%%%%%%%%%%%%%%%%%%%%%%%%%%%%
\section{A compact ancestry labeling scheme  for $\cF(n,d)$}
%%%%%%%%%%%%%%%%%%%%%%%%%%%%%%%%%%%%%%%

This section is devoted to proving the existence of an ancestry  labeling scheme of size $\log n +2\log d +O(1)$
for the family of rooted forests in $\cF(n,d)$.
Informally, the scheme performs as follows. We construct a set of intervals $U$ such that
the nodes of any forest in $\cF(n,d)$ can be mapped to  $U$, in a way that ancestry relation can be answered using a simple interval containment test.
I.e., we make sure that $u$ is an ancestor of $v$ in some forest $F$ if and only if the interval associated with $u$ contains the interval associated
with $v$. We call such a mapping an {\em ancestry mapping}.
A label of a node in $F$ is simply a pointer to an element in $U$, and thus can be encoded using $\log |U|$ bits.
Therefore, to get short labels we need $U$ to be small.

The construction of $U$ is done by induction on the number of nodes in the forest. Assume that 
there exists some set of intervals  $U_k$, such that for any forest of size at most $2^k$, there exists an ancestry mapping from $F$ to $U_k$,
and consider now the set of forests $\cF_{k+1}$ with at most $2^{k+1}$ nodes. Of course, if every $F\in \cF_{k+1}$ would break nicely into two forests with at most $2^k$ nodes each, then one could embed the two parts separately on two interval sets $U'$ and $U''$ of the same size as $U_k$. If that was always the case, we would ultimately get
an interval set $U=U_{\log n}$ of linear size for which any forest of size at most $n$ could be embedded to $U$ via an ancestry mapping,
and that would yield an ancestry labeling scheme with label size $\log n$.

Fortunately, life is not so simple, and a forest $F\in \cF_{k+1}$ doesn't always break nicely to two equal size sub-forests.
Specifically, problems occur whenever one must break a tree $T$ of $F$ into two parts and embed one part in $U'$ and the other part in $U''$. Ideally, if $F$ is broken into $F'\cup F''\cup T$ where $F'$ is embedded in $I'\subseteq U'$, and $F''$ is embedded in $I''\subseteq U''$, then one wants to embed $T$ by borrowing what remains free in $U'\setminus I'$ and $U''\setminus I''$. This can be achieved by using various scales of sub-interval sizes, so that to embed $T$ in $J=J'\cup J''$ with $J'\subseteq U'\setminus I'$ and $J'' \subseteq U''\setminus I''$.

Two difficulties arise in this recursive approach. The first one is related to the scale of the sub-intervals in which one picks $J'$ and $J''$. Indeed, too many sub-intervals yields  too many intervals in $U_{k+1}$. On the other hand, too few sub-intervals yields too large gaps between $I'$ and $J'$ in $U'$. This prevents the intervals in $U_{k+1}$ from being sufficiently compressed, and thus also ultimately results with too many intervals in $U_{k+1}$. Determining a good tradeoff between  the amount of scaling in the sub-intervals, and
 the gaps between intervals, in thus one major issue.

The second difficulty that is faced by the recursive approach is that splitting a tree into subtrees of sizes at most half is performed by removing the separator of the tree. However, one can see that whenever a tree $T$ of $2^{k+1}$ nodes is split into a collection $T_1,\dots,T_\ell$ of subtrees by removing the separator of $T$, the subtree containing the root of $T$ plays a special role in the setting of the ancestry scheme. Dealing with this special subtree is a second important issue, for which the assumption
on the depth of the forests will play a major role. The proof of the theorem below shows how to overcome these two issues.

\begin{theorem}\label{thm:anc}
There exists an ancestry labeling scheme  for the family of rooted forests in $\cF(n,d)$ whose label size is $\log n +2\log d +O(1)$.
% Moreover, the query time of the scheme is constant, and the time to construct the scheme for a given forest is linear.
\end{theorem}

\begin{proof}
For simplicity, we assume $n$ is a power of 2. (If $n$ is not a power of 2, we just round it to the next power of 2, say $N$, and we add $N-n$ independent nodes to the forest).
We begin by defining a set $U=U(n,d)$ of integers, which we use later to label
all forests in $\cF(n,d)$.

Let $c_0=1$, and, for any $i$, $1\leq i\leq \log n$,
let $$c_i=c_{i-1}+1/i^2 = 1+\sum_{j= 1}^i 1/{j^2}~.$$
We have $1+\sum_{j\geq 1} 1/{j^2}\leq 3$, and hence all the $c_i$'s are bounded from above by~3.
For any $i$, $1\leq i\leq \log n$, let us define the following values, that will be used to decompose integers: 
\begin{eqnarray*}
H_i & = & 1+ 3\cdot n \cdot d \cdot i^2/2^{i-1}\\
J_i & = & 2 \cdot d \cdot c_i \cdot i^2
 \end{eqnarray*}
 Then we define $\Gamma_0=3n$,  and
\[  \Gamma_i = \Gamma_0+\sum_{j=1}^i H_j \cdot J_j. \]
The set of integers $U$ is defined as the interval
\[U=[1,\Gamma_{\log n}).\]
Note that since $\Gamma_{\log n}=O(nd^2)$, we have $|U| =O(nd^2)$.
The marker algorithm maps the nodes of any forest $F\in\cF(n,d)$ into the integer set $U$.
To perform, the decoder algorithm represents each integer in $U$ as a unique triplet $(i,h,j)$, as follows.
\begin{itemize}
\item An integer  $\nu \in [1,\Gamma_0)$ is simply represented by $(0,\nu,0)$;
\item An integer $\nu$ that satisfies
$\Gamma_{i-1}\leq \nu < \Gamma_i$ for some $1\leq i \leq \log n$ can be described as
$$\nu=\Gamma_{i-1}+hJ_i +j$$ for unique $h$ and $j$ such that
 $h\in [0,H_i)$ and $j\in [0, J_i)$; Hence we represent such $\nu$ by the triplet $(i,h,j)$.
 \end{itemize}
% Observe that given an integer in $U$, one can easily retrieve its triplet representation in constant time. This can be done with the help of a data structure returning, for any integer $\nu$, $1\leq \nu \leq O(nd^2)$, the index $i$ such that $\Gamma_{i-1}\leq \nu < \Gamma_i$. (Alternatively, one could add $\log\log n$ bits to the label for encoding this index $i$).

For simplicity of presentation, in the following, we will not distinguish between an integer in $U$ and its triplet representation,
unless it may cause a confusion. Every integer in $U$ is associated with an interval as follows.
Let $x_0=1$, and for any $i$, $1\leq i\leq \log n$, let
$$x_i= \left\lceil\frac{2^{i-1}}{di^2}\right\rceil.$$
For $h\in [0,\Gamma_0)$, we associate  the triplet $(0,h,0)\in U$ with the interval
$ I_{0,h,0}=[h]$.
For any $i$, $1\leq i\leq \log n$, any $h\in [0,H_i)$, and any $j\in [0, J_i)$, we associate  the triplet $(i,h,j)\in U$ with the interval
$$ I_{i,h,j}=[x_ih,\; x_i(h+j)).$$

We now define a concept of specific interest for the purpose of our proof:

\begin{definition}
Let $F\in\cF(n,d)$. We say that a mapping $L:F\rightarrow U$ is an {\em ancestry mapping} if,
for every two nodes $u,v\in F$ with $L(u)=(i,h,j)$ and $L(v)=
(i',h',j')$, we have
\[\mbox{$u$ is an ancestor of $v$ in $F$} \iff  I_{i',h',j'}\subseteq I_{i,h,j}.\]
\end{definition}

In order to show that there exists an ancestry mapping from every forest in $\cF(n,d)$
into $U$, we shall make use of the following definitions.
For any interval $I\subseteq[1,\Gamma_0)$, let
\[U_0(I)=\{(0,\nu,0)\mid  \nu\in I\}\] and, for any $k$, $1\leq k \leq \log n$,
let $$U_k(I)=U_0(I)\cup \left\{(i,h,j)\mid 1\leq i \leq k ,~ h\in [0,H_i),~  j\in [0, J_i) \; ~\mbox{and}~ \; I_{i,h,j} \subseteq I\right\}.$$

The following observations are immediate by the definition of the sets $U_k(I)$.
Let $I$ and $J$ be two intervals in $[1,\Gamma_0)$. For any $k$, $1\leq k \leq \log n$, we have:
\begin{itemize}
\item
$I\cap J=\emptyset \; \Rightarrow \; U_k(I)\cap U_k(J)=\emptyset$,
\item
 $U_k(I)\cup U_k(J)\subseteq U_k(I\cup J)$,
 \item
$I\subset J \; \Rightarrow \; U_k(I)\subset U_k(J)$,
\item
$U_{k-1}(I)\subset U_k(I)$.
\end{itemize}

Fix $k$ such that $0\leq k\leq\log n$. 
We now give a sufficient condition for the existence of an ancestry labeling scheme using labels in $U_k(I)$. 
Let $I$ be an interval
in $[1,\Gamma_0)$ and let
$I_1,I_2, \cdots, I_t$  be a partition of $I$ into $t$ disjoint intervals, i.e.,
 $I=\cup_{i=1}^t I_i$ with $I_i\cap I_j=\emptyset$ for any $1\leq i < j \leq t$. Let $F$ be a forest, and let
$F_1,F_2, \cdots, F_t$ be $t$ pairwise disjoint forests such that $\cup_{i=1}^t F_i=F$. Using the four properties listed above, one can easily prove the following.

\begin{claim}\label{claim:obs2}
If there exists an ancestry mapping from $F_i$ to $U_k(I_i)$ for every $i$, $1\leq i\leq t$, then there
exists an ancestry mapping from $F$ to $U_k(I)$.
\end{claim}

The following  is the main technical ingredient for proving the theorem.

\begin{claim}\label{claim:main}
For every $k$, $0\leq k\leq \log n$, every forest $F$ of size $|F|\leq 2^k$ with depth bounded by $d$, and every interval $I\subseteq[1,\Gamma_0)$,
such that $|I|=\left\lfloor c_k|F|\right\rfloor$,
there exists an ancestry mapping of $F$ into $U_k(I)$.
\end{claim}

We prove this claim by induction on $k$. The claim for $k=0$ holds trivially.
Assume now that the claim holds for $k$ with $0\leq k< \log n$, and let us show that it also
holds for $k+1$.

Let $F$ be a forest  of size $|F|\leq 2^{k+1}$, and let $I\subseteq[1,\Gamma_0)$ be an interval, such that $|I|=\left\lfloor c_{k+1}|F|\right\rfloor$. Our goal is to show that
there exists an ancestry mapping of $F$ into $U_{k+1}(I)$. We consider two cases.

The simpler case is when all the trees in $F$ are of size at most
$2^{k}$. For this case, we show a claim stronger than what is stated in Claim~\ref{claim:main}. Specifically, we show that
there exists an ancestry mapping of $F$ into $U_k(I)$ for every interval $I\subseteq[1,\Gamma_0)$
such that $|I|=\left\lfloor c_k|F|\right\rfloor$ (i.e., a fraction $1/(k+1)^2$ of $|F|$ smaller than what is required to prove the claim).
Let $T_1,T_2,\cdots T_\ell$ be the trees in $F$.
We  divide the given interval  $I$ of size $\left\lfloor c_k|F|\right\rfloor$ into $\ell+1$ disjoint subintervals $I=I_1\cup I_2\cdots \cup I_\ell \cup I'$,
where $|I_i|=\left\lfloor c_k|T_i|\right\rfloor $ for every $i$,  $1\leq i\leq \ell$.
This can be done because $\sum_{i=1}^{\ell} \left\lfloor c_k|T_i|\right\rfloor \leq \left\lfloor c_k|F|\right\rfloor=|I|$.
By the induction hypothesis, we have an ancestry mapping of  $T_i$ into $U_k(I_i)$ for every $i$, $1\leq i\leq \ell$.
The stronger claim thus follows in this case by Claim~\ref{claim:obs2}.

Now consider the more involved case in which one of the subtrees in $F$, denoted by $T^*$, contains more then $2^{k}$ nodes.
Our goal now is to show that for every interval $I^*\subseteq[1,\Gamma_0)$, where $|I^*|=\left\lfloor c_{k+1}|T^*|\right\rfloor$,
there exists an ancestry mapping of $T^*$ into $U_{k+1}(I^*)$.
Once we show this, we can, similarly to the first case, divide the interval $I$ into $3$ disjoint subintervals $$I=I^*\cup I'\cup I'',$$
where $$|I^*|=\left\lfloor c_{k+1}|T^*|\right\rfloor \;\; \mbox{and} \;\; |I'|=\left\lfloor c_k|F'|\right\rfloor$$
with $F'=F\setminus T^*$.
Since we have an ancestry mapping that maps $T^*$ into $U_{k+1}(I^*)$, and one that maps
$F'$ into $U_{k}(I')$,
we get the desired ancestry mapping of $F$ into $U_{k+1}(I)$ by Claim~\ref{claim:obs2}. (The ancestry mapping of $F'$ into $U_{k}(I')$ can be done by the induction hypothesis, because $|F'|\leq 2^k$).

For the rest of the proof, our goal is thus to prove the following claim:
for every tree $T$ of size $|T|$ with $2^k<|T|\leq 2^{k+1}$, and every interval $I\subseteq[1,\Gamma_0)$, where $|I|=\left\lfloor c_{k+1}|T|\right\rfloor$,
there exists an ancestry mapping of $T$ into $U_{k+1}(I)$.

Recall that a {\em separator} of a tree $T$ is a node $v$ whose removal from $T$ (together with all its incident edges)
brakes $T$ into subtrees, each of size at most $|T|/2$.
It is a well known fact that every tree has a separator.
Note however, that there can be more than one separator to a  tree. Nevertheless, if this is the case then there are in fact two separators,
and one is the parent of the other. In the following, whenever we consider a separator of a rooted tree $T$, we refer
only to the separator of $T$ which is closer to the root.

We make use of  the following decomposition of $T$.
We  refer to  the  path $S$  from the separator of $T$ to the root of $T$
 as the {\em spine} of $T$. This spine may consist of only one node, namely, the root of $T$.
 Let $v_1,v_2,\cdots,v_{d'}$ be the nodes of the spine $S$, ordered bottom-up, i.e.,  $v_1$ is the separator of $T$ and $v_{d'}$ is the root of $T$. By this definition, we have that if $1\leq i<j\leq d'$ then $v_j$ is an ancestor of $v_i$.
A separator is not a leaf if $|T|>1$, and therefore $1\leq d' < d$. (Recall that the depth is 1 plus the distance to the root).
 By removing the nodes in the spine (and the edges connected to them),
 the tree $T$ brakes into $d'$ forests $F_1, F_2,\cdots, F_{d'}$, such that the following holds for each $1\leq i\leq d'$:

%{\bf Amos: need a picture to illustrate the decomposition}

 \begin{itemize}
 \item
 in $T$, the roots of the trees in $F_i$ are connected to $v_i$;
 \item
 each tree in $F_i$ contains at most $2^k$ nodes.
 \end{itemize}

The given interval $I$ for which we want to embed $T$ into $U_{k+1}(I)$ can be expressed as $I=[a,b)$ for some integers $a$ and $b$, and we have
$$b-a=|I|=\left\lfloor c_{k+1}|F|\right\rfloor.$$ For every  $i=1,\dots,d'$, we now define an interval $I_i$
(later, we will map $F_i$ into $U_{k}(I_i)$).
Let us first define $I_1$. Let $h_1$ be the smallest integer such that $a\leq h_1x_{k+1}$, and
let $\bar{h}_1$ be the smallest integer such that $\left\lfloor c_{k}|F_1|\right\rfloor\leq \bar{h}_1 x_{k+1}$. Note that $\bar{h}_1\geq 1$.
We let
 $$I_1=[h_1x_{k+1},(h_1+\bar{h}_1)x_{k+1}).$$ Assume now that we have defined
the interval $$I_i=[h_i x_{k+1},(h_i +\bar{h}_i)x_{k+1})$$ for $1\leq i< d'$. We define
the interval $I_{i+1}$ as follows. Let $h_{i+1}=h_i+\bar{h}_i$ and
let $\bar{h}_{i+1}$ be the smallest integers such that $\left\lfloor c_{k}|F_{i+1}|\right\rfloor\leq \bar{h}_{i+1}x_{k+1}$.
We let
 $$I_{i+1}=[h_{i+1}x_{k+1},(h_{i+1}+\bar{h}_{i+1})x_{k+1}).$$

Observe that for $1\leq i\leq d'$, the interval $I_{i}$ is simply  $I_{k+1,h_i,\bar{h}_i}$.
Note also that for every $i=1,\dots, d'$, we have
$$\bar{h}_i x_{k+1} <\left\lfloor c_{k}|F_i|\right\rfloor + x_{k+1}.$$ It follows that the
size of $I_i$ at most $\left\lfloor c_{k}|F_i|\right\rfloor + x_{k+1}-1$. Thus, since $h_1x_{k+1}<a + x_{k+1}$, we get that
\begin{eqnarray*}
\bigcup_{i=1}^{d'} I_i & \subseteq & \Big [a, \; a+(d'+1)(x_{k+1}-1) + \left\lfloor c_{k}|T|\right\rfloor \Big)\\
 & \subseteq & \Big [a, \; a+d\cdot (x_{k+1}-1)+ \left\lfloor c_{k}|T|\right\rfloor \Big).
\end{eqnarray*}
Now, since $d\cdot (x_{k+1}-1)\leq \left\lfloor\frac{2^k}{(k+1)^2}\right\rfloor$, and $2^k<|T|$, it follows that,
\begin{eqnarray*}\label{eq:1}
\bigcup_{i=1}^{d'} I_i & \subseteq & \left[a,a + \left\lfloor\frac{|T|}{(k+1)^2}+c_{k}|T|\right\rfloor\right) \\
 & = & [a,a + \left\lfloor c_{k+1}|T|\right\rfloor)\\
  & =& I.
\end{eqnarray*}
On the other hand, note  that for $1\leq i\leq d'$, $I_i$ contains at least $\left\lfloor c_{k}|F_i|\right\rfloor$ nodes. Therefore,
 by the fact that, for any $i$, $1\leq i\leq d'$,
 each tree in $F_i$ contains at most $2^k$ nodes, we get that
 there exists an ancestry mapping of each
$F_i$ into $U_k(I_i)$. We therefore  get an ancestry mapping from all $F_i$'s to $U_k(I)$, by Claim~\ref{claim:obs2}.
It is now left to map the nodes in the spine $S$ into $U_{k+1}(I)$, in a way that will respect the ancestry relation.

For every $i$, $1\leq i\leq d'$, let $\widehat{h}_i=\sum_{j=1}^i \overline{h}_j$.
We map the node $v_i$ of the spine to the triplet $(k+1,h_1,\widehat{h}_i)$.

Let us now show that $(k+1,h_1,\widehat{h}_i)$ is
in $U_{k+1}(I)$. First, the fact that $I_{k+1,h_1,\widehat{h}_i}\subseteq I$
follows from the fact that  $I_{k+1,h_1,\widehat{h}_i}=\cup_{j=1}^i I_j$, and using $\bigcup_{j=1}^{d'} I_j \subseteq I$.
It remains to show that $h_1\in [0,H_{k+1})$ and that $ \widehat{h}_i \in [0,J_{k+1})$. Note that,
$$a< 3n\leq \frac{3nd(k+1)^2}{2^k} \left\lceil\frac{2^k}{d(k+1)^2}\right\rceil=(H_{k+1}-1)x_{k+1}.$$
Therefore, the smallest integer $h_1$ such that $a\leq h_1x_{k+1}$
must satisfy $h_1\in [0,H_{k+1})$.
Recall now that for every $i$, $1\leq i\leq d'$, $\bar{h}_i$ is the smallest integer such that
$\left\lfloor c_{k}|F_i|\right\rfloor \leq \bar{h}_i x_{k+1}$. Thus
$$\bar{h}_i-1 <\frac{\left\lfloor c_{k}|F_i|\right\rfloor}{x_{k+1}}.$$
It follows that,
$$ \sum_{j=1}^i (\bar{h}_j -1) <
\sum_{j=1}^i \frac{\left\lfloor c_{k}|F_j|\right\rfloor}{x_{k+1}}  \leq
\frac{\left\lfloor c_{k}|F|\right\rfloor}{x_{k+1}}  \leq
\frac{ c_{k}2^{k+1}}{x_{k+1}}\leq 2c_kd(k+1)^2.$$
Thus
$$ \sum_{j=1}^i \bar{h}_j <d+2c_kd(k+1)^2 \leq 2c_{k+1}d(k+1)^2=J_{k+1}.$$
Therefore $ \widehat{h}_i \in[0,J_{k+1})$.

We now show that our mapping is indeed an ancestry mapping.
Observe first that, for $i$ and $j$ such that $1\leq i<j\leq d'$, we have $$I_{k+1,h_1,\widehat{h}_i}\subset I_{k+1,h_1,\widehat{h}_j}.$$
Thus, the interval associated with $v_j$ contains the one associated with $v_i$, as desired.

In addition,
recall  that, for every $i=1,\dots,d'$, $F_i$ is mapped into $U_k(I_i)$. Therefore, if $L(v)$ is the triplet
of some node $v \in F_i$, then the interval associated with it is contained in $I_i$. Since
 $I_i\subset I_{k+1,h_1,\widehat{h}_j}$ for every $j$ such that $1\leq i<j\leq d'$, we obtain
 that the interval associated with $v$ is contained in the interval associated with $v_j$.
This completes the proof of Claim~\ref{claim:main}.

From Claim~\ref{claim:main},  we get that  there exists an ancestry mapping of any $F\in\cF(n,d)$ into $U$.
We use this ancestry mapping to label
the nodes in $F$: an ancestry query between two labels
can be answered using a simple interval containment test between the corresponding intervals.
% we get that each such query can be performed in constant time.
The stated label size follows, as each node can be encoded using $\log |U|$ bits,
and $|U|=\Gamma_{\log n}= O(nd^2)$.
% The fact that the mapping of a given forest to $U$ takex $O(n)$ time follows from its description.
This completes the proof of the theorem.
\end{proof}

%%%%%%%%%%%%%%%%%%%%%%%%%%%%%%%%%%%%%%%
\section{A compact adjacency labeling scheme and a small universal graph for $\cF(n,d)$}
%%%%%%%%%%%%%%%%%%%%%%%%%%%%%%%%%%%%%%%

The ancestry labeling scheme described in the previous section can be advantageously transformed into an adjacency labeling scheme for trees of bounded depth. Recall that an {\em adjacency labeling
scheme}  for the  family of graphs $\cG$ is a pair $(\cM,\cD)$ of marker and decoder, satisfying that if $L(u)$ and $L(v)$ are the labels given 
by the marker $\cM$ to
two nodes $u$ and $v$ in some graph $G\in\cG$, then
\[
\cD(L(u),L(v))=1 \iff \mbox{$u$ and $v$ are adjacent in $G$.}
\]
Similarly to the ancestry case, we evaluate an adjacency labeling scheme
$(\cM,\cD)$ by its {\em label size}, namely
the maximum number of bits in a label assigned by the marker
algorithm $\cM$ to any node in any graph in $\cG$.

For any two nodes $u$ and $v$ in a rooted forest $F$, $u$ is a parent of $v$ if and only if  $u$ is an ancestor
of $v$ and $depth(u)=depth(v)-1$. Also, $u$ is a neighbor of $v$ if and only if either $u$ is a parent of $v$ or
$v$ is a parent of $u$. It therefore follows that one can easily transform any ancestry labeling scheme for
$\cF(n,d)$ to an adjacency labeling scheme for
$\cF(n,d)$ with an extra additive term of $\log d$ bits to the label size (these bits are simply used to encode
the depth of a vertex).
Using Theorem \ref{thm:anc} we thus obtain  the following.

\begin{theorem}\label{thm:adj}
There exists an adjacency labeling scheme for $\cF(n,d)$ of size $\log n +3\log d +O(1)$.
%Moreover, the query time of the scheme is constant and the time to construct the scheme for a given forest is linear.
\end{theorem}

Interestingly enough, this latter adjacency labeling scheme enables to give a short implicit representation (in the sense of \cite{KNR92}) of all forests
 with bounded depth. Recall that a graph $G$ is an {\em induced subgraph} of a graph $\cU$  if there exists a one-to-one (but not necessarily onto) mapping $\phi$ from $V(G)$ to $V(\cU)$ such that
\[\forall u,v\in V(G),\;\;\;\;  \{u,v\}\in E(G) \iff \{\phi(u),\phi(v)\}\in E(\cU).\]
Given a graph
family $\cG$, a graph $\cU$ is {\em universal} for $\cG$  if every
graph in $\cG$ is an induced subgraph of $\cU$. Note that a variant of this notion considers the graph $\cU$ as universal for  $\cG$ whenever every graph in $\cG$ is a partial subgraph of $\cU$, i.e., the existence of an edge between $\phi(u)$ and $\phi(v)$ in $E(\cU)$ does not necessarily imply the existence of the edge $\{u,v\}$. This variant enables to analyze universal graphs for infinite graph classes~\cite{R64}. The notion of universality considered in this paper is somewhat more restrictive, but it enables to relate the size of a universal graph for $\cG$ with the size of the graphs in $\cG$. Moreover, this notion of universality precisely captures the structure of the graphs in $\cG$. In fact,  there is a tight relation between this notion and adjacency labeling schemes:

\begin{lemma} \label{naor}{\rm (S.~Kannan, M.~Naor, and S.~Rudich \cite{KNR92})}\\
A graph family $\cG$ has an adjacency labeling scheme with label
size $k$ if and only if there exists a universal graph for $\cG$, with
$2^{k}$ nodes.
\end{lemma}

Combining the lemma above with Theorem \ref{thm:adj}, we get the corollary below.

\begin{corollary}\label{cor:universal}
Let $d$ be a constant integer. There exists a universal graph for $\cF(n,d)$, with
$O(n)$ nodes.
\end{corollary}

Proving or disproving the existence of a universal graph with a linear number of nodes for the class of $n$-node trees is a central open problem in the design of informative labeling schemes.

\newpage

%%%%%%%%%%%%%%%%%%%%%%%%%%%%%%%%%%%%

\end{document}